\documentclass[11pt]{article}

\usepackage[preprint]{acl}

\usepackage{times}
\usepackage{latexsym}
\usepackage{amsmath}
\usepackage{amsthm}
\usepackage{amssymb}
\usepackage{multicol}
\usepackage{multirow}
\usepackage{subcaption}
\usepackage[T1]{fontenc}
\usepackage{booktabs}
\usepackage[utf8]{inputenc}

\usepackage{microtype}

\usepackage{inconsolata}

\usepackage{graphicx}
\newtheorem{assumption}{Assumption}

\newtheorem{theorem}{Theorem}

\usepackage{chngcntr}
\newtheorem{lemma}{Lemma}
\title{A Distribution-Free Framework for Rewrite-Based Human-text Detection via Knockoff Filtering}

\author{Yi Liu \\
   Prorata.ai\\
  \texttt{yi@prorata.ai} \\
}

\begin{document}
\maketitle
\begin{abstract}
We propose a distribution-free statistical framework that converts arbitrary rewrite-based detectors into detectors with finite-sample FDR guarantees without retraining. Our key observation is that rewrite-based detection implicitly constructs knockoff samples, enabling LLM-generated text detection to be formulated as a multiple hypothesis testing problem with knockoff structure. This perspective separates the design of detection statistics from the control of false discoveries, allowing existing rewrite detectors to inherit finite-sample false discovery rate (FDR) guarantees through a simple calibration procedure. We demonstrate reliable FDR control with meaningful detection power across three detection models, 19 domains, and four LLMs. 
\end{abstract}

\section{Introduction}

Detecting LLM-generated text has attracted significant recent attention. Existing approaches generally fall into two categories: logit-based methods, which analyze token probability statistics, and rewrite-based methods, which compare a text with an LLM-generated rewrite \cite{gehrmann-etal-2019-gltr,DBLP:conf/nips/HuCH23,solaiman2019release,zhou2026learn, maoraidar, NEURIPS2024_1d35af80,zhu-etal-2023-beat,zhou2025adadetect,hans2024spotting,bao2023fast,su-etal-2023-detectllm,chen2025imitate}.

We observe that rewrite-based detection implicitly constructs text--rewrite pairs that are approximately exchangeable under AI-generated text, naturally inducing a knockoff structure. Under this perspective, rewrite-based detection can be reformulated as a multiple hypothesis testing problem with finite-sample false discovery rate (FDR) guarantees. Building on this observation, we introduce a distribution-free statistical framework that converts arbitrary rewrite-based detectors into detectors with finite-sample FDR guarantees without retraining.
Our contributions are threefold:\\
1. We show that rewrite-based detection implicitly generates knockoff samples, enabling a multiple hypothesis testing formulation.\\
2. We introduce a distribution-free detection procedure that converts rewrite-based statistics into detectors with provable FDR control under a knockoff symmetry condition.\\
3. We provide theoretical and empirical results showing that the proposed framework controls the FDR whenever the knockoff symmetry condition holds.

\subsection{Rewrite-based approaches}

Rewrite-based methods compare a text with an LLM-generated rewrite. For AI-generated text, the original text and its rewrite are approximately exchangeable, since both arise from the same underlying LLM distribution. In contrast, rewriting human text induces a directional shift toward the LLM distribution \cite{zhou2026learn,tulchinskii2023intrinsic}. This asymmetry naturally connects rewrite-based detection to knockoff filtering.

In many NLP applications, LLM detection is not used to classify a single isolated document, but to screen large corpora such as reviews, assignments, web pages, or benchmark data. In corpus-scale detection, threshold selection matters as much as score quality. A detector with high AUROC may still produce many false accusations when applied to thousands of documents. FDR control gives a deployment-facing guarantee: among texts classified as human-written, the expected fraction that are actually AI-generated is bounded by a user-chosen level.  Moreover, many approaches require retraining whenever the rewrite model changes, making them costly and brittle in practice.

\subsection{Knockoff Filter}
We treat AI-generated text as the null hypothesis and human-written text as the alternative; discoveries therefore correspond to texts selected as human-written.
Treating this as a hypothesis testing problem and solving it using knockoff filters offer three principal advantages \cite{10.1111/rssb.12265,barber2015controlling}. First, knockoff methods are intrinsically distribution-free, eliminating the need to explicitly train a classifier. Second, knockoff methods provide strong theoretical guarantees—--an attribute that remains relatively uncommon in NLP literature. Third, the framework can extend many logit-based detection methods to the rewrite setting. This construction enables existing logit-based detectors to obtain FDR control.
\section{Methodology}
\label{sec:Methods}
Our framework consists of three steps:\\
1. Generate a knockoff by rewriting.\\
2. Compute a comparison statistic between the original and knockoff.\\
3. Apply the knockoff filtering rule to select texts classified as human-written while controlling FDR.
\subsection{Generate Knockoffs}
We consider a collection of texts $T_1, T_2, \ldots, T_n$, among which $N_0$ is generated by AI. We assume that $T_1, T_2,\cdots T_n$ are independently generated. The objective is to identify human-written texts. Let $H_i$ denote the indicator variable that equals $1$ if $T_i$ is generated by a human and $0$ otherwise. For each text $T_i$, we prompt an LLM to generate a corresponding rewrite $R_i$. Prior work \cite{zhou2026learn, maoraidar, NEURIPS2024_1d35af80,zhu-etal-2023-beat} introduces a statistic $s_i = f(T_i, R_i)$, which serves as a measure of distance between $T_i$ and  $R_i$. Both theoretical and empirical results in \cite{zhou2026learn} demonstrate that larger values of $s_i$ are associated with a higher likelihood that $T_i$ was written by a human rather than generated by an LLM. We denote the resulting collection of statistics by $\mathcal{S} = {s_1, \ldots, s_n}$.

To prove our theoretical results, we require the standard knockoff sign-flip (symmetry) condition. Lemma \ref{lemma:exchangability_condition} gives a sufficient condition under which it holds for rewrite pairs. Specifically,
\begin{assumption}
\label{assumption:key_assumption}
$\mathbb{P}(f(T_i, R_i) < - a\mid H_i =0) = \mathbb{P}(f(T_i,R_i) > a| H_i = 0)$ for all $a \in \mathbb{R}^+$. 
\end{assumption}
This symmetry arises from an exchangeability property of rewrite pairs under the null hypothesis that the text is AI-generated.

\subsection{Finding Threshold}
In the final step, we want to generate the threshold using $ \mathcal{S}$. We let $q$ be the target FDR. The threshold $\tau$ is defined as 
\begin{equation}
\label{eqn:threshold}
    \tau = \min_{c \in \mathbb{R}^+} \left\{ \frac{\#\{ s_i \leq -c  \}+1}{\max\{\#\{s_j \geq c \}, 1\}} \leq q \right\}
\end{equation}
 larger values of $s_i$ indicate human-written text, we classify $T_i$ as human-written whenever $s_i > \tau$.
\section{Theoretical Analysis}
\label{sec:theory}
The following lemma provides a sufficient condition under which Assumption \ref{assumption:key_assumption} holds. This result is similar to definition 2 in \cite{10.1111/rssb.12265} which inspires exchangeability conditions. The proof is provided in Appendix \ref{appendix:proofs}.
\begin{lemma}
\label{lemma:exchangability_condition}
Let $R_i \sim K(\cdot\mid T_i)$ be a stochastic kernel that generates rewrites for $T_i$ and $\mathbb{P}_{\text{LLM}}(\cdot)$ be the LLM for which $T_i$ is sampled from. If 
$$\mathbb{P}_{\text{LLM}}(T_i)K(R_i\mid T_i) =  \mathbb{P}_{\text{LLM}}(R_i)K(T_i\mid R_i),$$
Then for any anti-symmetric statistics,
$$f(T_i, R_i) = -f(R_i, T_i),$$
$$ \mathbb{P}(f(T_i, R_i) < - a) = \mathbb{P}(f(T_i,R_i) > a) , \forall a \in \mathbb{R}^+$$ 
\end{lemma}
When Assumption 1 holds only approximately—as is typical in practice since no rewriter exactly satisfies the exchangeability condition of Lemma 1—exact finite-sample FDR control is no longer guaranteed, so we monitor symmetry empirically using the frac+ diagnostic reported in Table \ref{tab:symmetry}.
 
\begin{theorem} 
\label{theorem:main}
Under Assumption 1, $\mathcal{S}_\tau = \{ i : s_i > \tau \}$ controls the usual FDR,
$$ \text{FDR} = \mathbb{E} \left(\frac{\#\{i \in\mathcal{S}_\tau, H_i = 0 \}}{\max\{|\mathcal{S}_\tau|, 1\}} \right) \leq q.$$
The guarantee holds in finite samples. 
\end{theorem}
The proof follows the knockoff filtering argument of \cite{10.1111/rssb.12265, barber2015controlling}, adapted to the rewrite detection setting. The complete proof is presented in Appendix \ref{appendix:proofs}.

\section{Experiments}

\begin{figure*}[t]
  \centering
  \includegraphics[width=\linewidth]{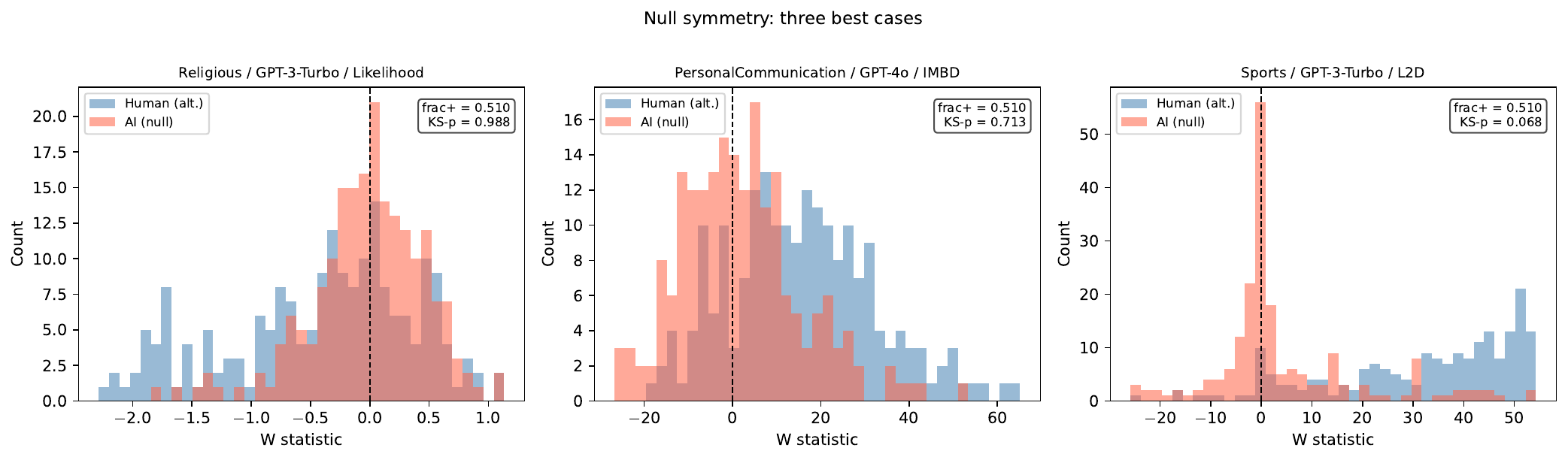}
  \caption{%
    Distribution of knockoff statistics $s_i$ for human (alternative, blue) and
    AI-generated (null, red) texts in representative high-symmetry cases,
    one per rewriting method: Religious/GPT-3.5T (Likelihood, frac$+$\,=\,0.510,
    KS-$p$\,=\,0.988), PersonalCommunication/GPT-4o (IMBD, frac$+$\,=\,0.510,
    KS-$p$\,=\,0.713), and Sports/GPT-3.5T (L2D, frac$+$\,=\,0.510,
    KS-$p$\,=\,0.068).
    The dashed vertical line marks $s_i = 0$.
    The near-symmetric null distribution around zero confirms that
    Assumption~1 holds in practice for these cases.
  }
  \label{fig:null_histograms}
\end{figure*}
\begin{table*}[t]
\small
\centering
\caption{Empirical validation of the symmetry condition (Assumption~1). We report the fraction of knockoff statistics $s_i > 0$ (frac$+$) and the KS test p-value for symmetry of $s_i$, averaged across domains. Values close to 0.5 for frac$+$ indicate the symmetry assumption holds for AI texts. The cross-domain rows show results when the null mean is borrowed from other domains rather than estimated within-domain.}
\label{tab:symmetry}
\begin{tabular}{lcccccccc}
\toprule
Method & \multicolumn{2}{c}{GPT-3.5T} & \multicolumn{2}{c}{GPT-4o} & \multicolumn{2}{c}{Gemini} & \multicolumn{2}{c}{Llama-70B} \\
\cmidrule(lr){2-3}\cmidrule(lr){4-5}\cmidrule(lr){6-7}\cmidrule(lr){8-9}
  & frac$+$ & KS-$p$ & frac$+$ & KS-$p$ & frac$+$ & KS-$p$ & frac$+$ & KS-$p$ \\
\midrule
L2D  & 0.308 & 0.000 & 0.302 & 0.006 & 0.312 & 0.024 & 0.287 & 0.002 \\
Likelihood  & 0.499 & 0.070 & 0.517 & 0.078 & 0.509 & 0.114 & 0.513 & 0.059 \\
IMBD  & 0.466 & 0.042 & 0.475 & 0.040 & 0.483 & 0.070 & 0.469 & 0.043 \\
\bottomrule
\end{tabular}
\end{table*}

\begin{table*}[t]
\small
\centering
\caption{Empirical FDR control and detection power averaged across domains when the null mean is borrowed from other domains (cross-domain transfer). The first row per method reports observed FDR and the second reports detection power at the corresponding target level $q$.}
\label{tab:fdr_power_cross}
\begin{tabular}{llccc|ccc}
\toprule
Method &  & \multicolumn{3}{c|}{GPT-3.5T} & \multicolumn{3}{c}{GPT-4o} \\
\cmidrule(lr){3-5}\cmidrule(lr){6-8}
 &  & $q=0.2$ & $q=0.3$ & $q=0.5$ & $q=0.2$ & $q=0.3$ & $q=0.5$ \\
\midrule
\multirow{2}{*}{L2D } & FDR & 0.173 & 0.183 & 0.203 & 0.169 & 0.183 & 0.201 \\
 & Power & 0.833 & 0.850 & 0.871 & 0.825 & 0.846 & 0.870 \\
\multirow{2}{*}{Likelihood } & FDR & 0.202 & 0.235 & 0.319 & 0.182 & 0.209 & 0.278 \\
 & Power & 0.220 & 0.260 & 0.316 & 0.284 & 0.321 & 0.375 \\
\multirow{2}{*}{IMBD } & FDR & 0.216 & 0.239 & 0.265 & 0.222 & 0.240 & 0.262 \\
 & Power & 0.682 & 0.704 & 0.732 & 0.649 & 0.670 & 0.699 \\
\bottomrule
\end{tabular}
\vspace{4pt}

\begin{tabular}{llccc|ccc}
\toprule
Method &  & \multicolumn{3}{c|}{Gemini} & \multicolumn{3}{c}{Llama-70B} \\
\cmidrule(lr){3-5}\cmidrule(lr){6-8}
 &  & $q=0.2$ & $q=0.3$ & $q=0.5$ & $q=0.2$ & $q=0.3$ & $q=0.5$ \\
\midrule
\multirow{2}{*}{L2D } & FDR & 0.157 & 0.172 & 0.189 & 0.160 & 0.168 & 0.184 \\
 & Power & 0.941 & 0.946 & 0.953 & 0.901 & 0.912 & 0.923 \\
\multirow{2}{*}{Likelihood } & FDR & 0.168 & 0.206 & 0.268 & 0.199 & 0.225 & 0.276 \\
 & Power & 0.440 & 0.499 & 0.578 & 0.471 & 0.529 & 0.610 \\
\multirow{2}{*}{IMBD} & FDR & 0.209 & 0.242 & 0.284 & 0.213 & 0.235 & 0.262 \\
 & Power & 0.903 & 0.921 & 0.944 & 0.790 & 0.808 & 0.834 \\
\bottomrule
\end{tabular}
\end{table*}
Our experiments evaluate two properties of the proposed framework: (1) whether knockoff filtering achieves empirical FDR control across domains and source models, and (2) how statistical power depends on the underlying rewrite detection statistic.

We evaluate the proposed knockoff framework using three detection statistics: ImBD \cite{chen2025imitate}, L2D \cite{zhou2026learn}, and likelihood \cite{gehrmann-etal-2019-gltr}. 
For ImBD and likelihood  we generate scores $g(T_i)$ and $g(R_i)$ for both text and its rewrite using the function estimator $g(\cdot)$ and denote $s_i = f(T_i,R_i) = g(T_i)- g(R_i)$ . For L2D, we use $f(T_i, R_i) = \frac{\log p_{\phi}(T_i)}{|T_i|} - \frac{\log p_{\phi}(R_i)}{|R_i|} $ (we remove the absolute value to obtain a signed statistic), which ensures that the resulting statistic is antisymmetric and centered around zero under the null hypothesis. 

Our evaluation is based  on the dataset released by \cite{hao-etal-2025-learning,zhou2026learn}, covering  19 diverse domains (Academic Research, Art \& Culture, Business, Education, Entertainment,
Environmental, Finance, Food, Government, Legal, Medical,  News, Online
Content, Product Review, Religious, Sports,
Technical Writing, and Travel) with 200 human-written and 200 LLM texts from GPT-3.5Turbo, GPT-4o, Gemini 1.5 Pro, and 
Llama-3-70B-Instruct. The rewrites are done by \texttt{gemma-9b-instruct}.  For L2D and IMBD, we negate the raw statistics 
as for these two methods, larger values consistently correspond to human authorship. 

\subsection{Symmetry Condition}
We assess the symmetry condition (Assumption~\ref{assumption:key_assumption}) using two diagnostics: the fraction of null statistics satisfying $s_{0i} > 0$ ($\text{frac}^+$), which should be close to $0.5$ under symmetry, and the KS test p-value for symmetry of the $s_i$ distribution.

To improve empirical symmetry, we consider a cross-domain transfer setting where the null mean is estimated from a held-out source domain rather than the target domain itself. Concretely, for each source domain we compute the mean
of the AI scores (the null distribution) after sign transformation,
and subtract this value from all scores in every other target domain
without any further centering. The knockoff filter is then applied
directly to these shifted scores. For each target domain, results are
averaged over all 18 source domains, and we report FDR, power,
fraction of positive knockoff statistics (frac$+$), and the KS
symmetry p-value. This setup isolates the contribution of the
demeaning step. 

\subsection{Empirical FDR Control}
We verify empirical FDR control and detection power by comparing 
observed FDR and power to the target level $q \in \{ 0.2, 0.3, 
0.5\}$ for each method and model, averaged across all domains 
(Table \ref{tab:fdr_power_cross}).

\paragraph{L2D and IMBD.} 
FDR ranges from $0.16$ to $0.28$ at $q = 0.2$ across all models and
methods and is slightly above the nominal level. At the same time, detection
power is strong. L2D achieves power of $0.833$--$0.941$ at $q = 0.2$ across models,
and IMBD reaches $0.649$--$0.903$, with both methods performing
strongest on Gemini and Llama-70B. The FDR inflation is fairly stable
across $q$ levels and models.

\paragraph{Likelihood.} For the likelihood method, the model is not trained to identify human-written text from AI written text. However, we gain substantial power with almost no FDR inflation. At $q = 0.2$, power rises to $0.220$--$0.440$ across models, and reaches
$0.471$--$0.529$ on Llama-70B and Gemini at $q = 0.3$. We observe slight FDR inflation for GPT-3.5T cases. This shows that the method performs even without additional training.

\paragraph{Summary.}Across 19 domains and four source models, cross-domain mean correction substantially improves power, while FDR remains close to the nominal level for L2D and is mildly inflated for some method–model pairs. FDR control and meaningful power are jointly 
achieved by L2D and IMBD. 

\section{Discussion}
The central contribution of this work is a statistical calibration framework for controlling false discoveries in existing rewrite-based detectors, rather than designing or training a new detector. We show that rewrite-based detection implicitly constructs knockoff samples, allowing arbitrary rewrite detectors to be converted into finite-sample calibrated detectors with FDR guarantees and without retraining. This separates the design of detection statistics from the control of false discoveries, allowing each component to be improved independently.

On the one hand, our theory proves that under symmetry, false discovery control is guaranteed. On the other hand, our experiments highlight that this methodology is generally robust under different detectors, even for detectors that are not specifically trained to detect LLM text, generating meaningful power and low false discovery rate. 
\section{Limitation}
The main limitation of the proposed framework is its reliance on the knockoff symmetry condition which may only approximately hold in empirical settings. The finite-sample FDR guarantee holds under Assumption 1, but real rewrite models do not exactly satisfy the required exchangeability condition. In practice, the signed statistics can exhibit systematic directional bias, so we apply a mean-correction step before thresholding. This makes calibration quality important: if the estimated centering constant does not adequately restore null symmetry, empirical FDR control may become conservative or anti-conservative.

Our cross-domain experiment illustrates this tradeoff. Borrowing the null mean from other domains substantially improves power, especially for L2D and IMBD, but also introduces mild FDR inflation at some target levels. This suggests that cross-domain calibration is useful when target-domain calibration data are limited, but it should be accompanied by empirical symmetry diagnostics such as frac+ and KS tests. More generally, the proposed method should be viewed as a calibration layer for rewrite-based detectors rather than a complete detector by itself.

A second limitation is statistical power. The knockoff filter controls false discoveries but cannot compensate for an uninformative base statistic. When the rewrite statistic weakly separates human-written and AI-generated text, the procedure may make few discoveries. Improving the underlying rewrite statistic and designing calibration procedures that preserve both symmetry and power remain important directions for future work.
\bibliography{custom}
\newpage
\onecolumn
\appendix
\section*{Appendix}
\counterwithin{table}{section}
\renewcommand{\thetable}{\thesection.\arabic{table}}
\label{appendix}
\section{Theorems and Proofs}

\label{appendix:proofs}
\subsection{Proof of lemma \ref{lemma:exchangability_condition}}
\begin{proof}
\begin{align*}
    &\mathbb{P}_{\text{LLM}}(T_i)K(R_i\mid T_i) =  \mathbb{P}_{\text{LLM}}(R_i)K(T_i\mid R_i)\\
    &\Leftrightarrow \mathbb{P}(T_i, R_i) = \mathbb{P}(R_i, T_i)\\
    &\Leftrightarrow \mathbb{P}(f(T_i, R_i) > a ) = \mathbb{P}(-f(R_i, T_i)> a) \quad \forall a\in \mathbb{R}^+\\
    &\Leftrightarrow \mathbb{P}(f(T_i, R_i) < - a) = \mathbb{P}(f(T_i,R_i) > a)\quad \forall a\in \mathbb{R}^+
\end{align*}
\end{proof}

\subsection{Proof of theorem \ref{theorem:main}}
\subsubsection{Main Proof of Theorem \ref{theorem:main}}
\begin{proof}
    For all $i$ such that $H_i =0$, let $\pi$ denotes a permutation function such that $|s_{\pi(1)}| > |s_{\pi(2)}| > ...|s_{\pi(N_0)}|$. Define a filtration $\mathcal{F}_k$ such that with a threshold $c$ move from $|s_{\pi(1)}|$ to $|s_{\pi(N_0)}|$, revealing $|s_{\pi(1)}|,...,|s_{\pi(k)}|$ and $\text{sign}(s_{\pi(1)}), ..., \text{sign}(s_{\pi(k)})$. 
    
    Under Assumption \ref{assumption:key_assumption}, conditional on this filtration $\mathcal{F}_{k-1}$,  and $|s_i|$, $\text{sign}(s_{\pi(k)})$  is i.i.d. Rademacher variable. 

    Let $V_k= \#\{ i\leq k: \text{sign}(s_{\pi(i)}) = +1 \}$ and $U_k = \#\{i\leq k: \text{sign}(s_{\pi(i)}) = -1\}$ so that $V_k +U_k = k$. Consider the ratio $M_k = V_k/(U_k+1)$ and using from result from Lemma \ref{lemma:supermartingale} , $M_k$ is a super-martingale by reverse filtration. So $\mathbb{E}(M_k) \leq \mathbb{E}(M_{N_0}) \leq 1$ by optional stopping time theorem. Using Lemma \ref{lemma:first_inequality} and Lemma \ref{lemma:second_inequality}, we have that
    $$ \text{FDR} \leq q \mathbb{E}(M_k) \leq q \mathbb{E}(M_{N_0}) \leq q.$$
\end{proof}
\subsubsection{Proof of Supporting Lemmas}
\begin{lemma}
\label{lemma:first_inequality}
if $$ 
    \tau = \min_{c \in \mathbb{R}^+} \left\{ \frac{\#\{ s_i \leq -c \}+1}{\#\{s_j \geq c \}} \leq q \right\},$$
then
    $$ \text{FDR} \leq q \mathbb{E}(M_k).$$
\end{lemma}
\begin{proof}
\begin{align*}
     \text{FDP} &=  \frac{\#\{ s_i \geq \tau, H_i = 0 \}}{\max\{\#\{s_j \geq \tau \},1\}} \\
     &\leq  \frac{\#\{ s_i \geq \tau, H_i = 0 \}}{\max\{\#\{s_j \geq \tau \},1\}} \times \frac{\#\{ s_i \leq -\tau \}+1}{\#\{s_j \leq- \tau,H_i = 0 \}+1}\\
     & = \frac{\#\{ s_i \leq -\tau \}+1}{\max\{\#\{s_j \geq \tau \},1\}} \times \frac{\#\{ s_i \geq \tau,H_i = 0 \}}{\#\{s_j \leq- \tau,H_i = 0 \}+1}\\
     &\leq q \times \frac{\#\{ s_i \geq \tau, H_i = 0 \}}{\#\{s_j \leq- \tau,H_i = 0 \}+1}\\
     &= q M_k.
\end{align*}
since, by definition, 
$$ 
    \tau = \min_{c \in \mathbb{R}^+} \left\{ \frac{\#\{ s_i \leq -c \}+1}{\#\{s_j \geq c \}} \leq q \right\}.
$$
Then 
$$\text{FDR} = \mathbb{E}(\text{FDP}) \leq q\mathbb{E}(M_k) $$
\end{proof}

\begin{lemma}
\label{lemma:supermartingale}
   For $k = N_0, N_0-1, N_0-2 ,\dots 1, 0$, let $V_k= \#\{ i\leq k: \text{sign}(s_{\pi(i)}) = +1 \}$ and $U_k = \#\{i\leq k: \text{sign}(s_{\pi(i)}) = -1\}$ with $V_0 = U_0  = 0$. Let $\mathcal{F}_k$ be the filtration defined by knowing all the non-null $s_i$ as $V_{k'}$ and $U_{k'}$ for all $k' \geq k$, then 
   $$ M_k = \frac{V_k}{U_k +1 }$$
   is a super-martingale running backward in time with respect to $\mathcal{F}_k$. 
\end{lemma}
\begin{proof}
    Note that the filtration $\mathcal{F}_k$ informs us about whether $k$ is null or not, since the non-null process is known exactly. On the one hand, if $k$ is non-null, then
    $M_k = M_{k-1}$ since $V_k$ and $U_k$ are counting nulls. On the other hand, if $k$ is null, then
    $$M_{k-1} = \frac{V_k-I}{\max \{U_k + I,1\}} \quad \text{where $I = \mathbb{I}_{s_k < c}$ }$$ 
So $\mathcal{F}_k$ gives no further knowledge about $I$, and it follows from the exchangeability property of the $s_k$— they are i.i.d. and thus exchangeable —that $\mathbb{P}(I = 1) = V_k/(V_k + U_k
)$. Thus in the case where $H_k=0$,
\begin{align*}
    \mathbb{E}(M_k\mid \mathcal{F}_k) &= \frac{1}{V_k+U_k}\left(V_k\frac{V_k +1 }{U_k - 1} + U_k \frac{V_k}{\max(U_k,1)}\right)\\
    &= \begin{cases}
        &\frac{V_k}{1+U_k}, \quad U_k >1\\
        &V_k-1, \quad U_k = 0
    \end{cases}
\end{align*}
Therefore, 
$$\mathbb{E}(M_{k-1} \mid \mathcal{F}_k) = \begin{cases}
M_k \quad \text{$s_k$ is non-null}\\
M_k \quad \text{$s_k$ is null $U_k >0$}\\
M_k -1 \quad \text{$s_k$ is null and $U_k =0$}
\end{cases}$$
which shows that $M_k$ is supermartingale.
\end{proof}
\begin{lemma}
\label{lemma:second_inequality}
    $\mathbb{E}(M_{N_0}) \leq 1$
\end{lemma}
\begin{proof}
    Notice that $V_{N_0} \sim \text{Binomial}(N_0,\frac{1}{2})$ and $U_{N_0} = N_0 - V_{N_0}$, then we have 
    \begin{align*}
    M_{N_0} &= \mathbb{E}\left(\frac{V_{N_0}}{N_0 - V_{N_0} +1}\right)\\
    & = \sum_{v=0}^{N_0} \frac{v}{N_0 - v+1} {N_0 \choose v}\left(\frac{1}{2}\right)^v\left(1- \frac{1}{2}\right)^{N_0-v}\\
    &= \sum_{v=0}^{N_0} \frac{v}{N_0 - v+1} \frac{N_0!}{v!(N_0-v)!}\left(\frac{1}{2}\right)^v\left(1- \frac{1}{2}\right)^{N_0-v}\\
    &= \sum_{v=0}^{N_0}  \frac{N_0!}{(v-1)!(N_0-v+1)!}\left(\frac{1}{2}\right)^v\left(1- \frac{1}{2}\right)^{N_0-v}\\
    &= \frac{\frac{1}{2}}{1-\frac{1}{2}}\sum_{v=0}^{N_0}  \frac{N_0!}{(v-1)!(N_0-v+1)!}\left(\frac{1}{2}\right)^{v-1}\left(1- \frac{1}{2}\right)^{N_0-v+1}\\
    &\leq 1
\end{align*}
\end{proof}
\newpage 
\twocolumn
\section{Experiment}
\subsection{Data Preparation}
\footnote{\tiny The code are in \texttt{https://anonymous.4open.science/r/ KnockoffIdentification-A2CC/README.md}} The data are from \cite{zhou2026learn, hao-etal-2025-learning}\footnote{\tiny\texttt{https://github.com/Mamba413/L2D/tree/main/exp\_diverse}}. The licenses for each of the data are shown in Table \ref{tab:dataset_details}. \footnote{\tiny We skipped Code and Creative Writing category because it was skipped in the cross-training for IMBD in \cite{zhou2026learn}.}Our experiments use publicly available benchmark and web datasets collected from existing research corpora and public sources, as documented in Appendix B.1.

We do not collect new personally identifying information (PII), nor do we attempt to identify individuals. The datasets are used in their original released form and are restricted to research and evaluation purposes. We follow the original dataset access conditions and use the data only for research evaluation.
\subsection{Rewrite Setup}
All three methods share a common rewriting backbone: texts are rewritten
using \texttt{gemma-9b-instruct} (top-$p{=}0.96$, temperature$=0.7$,
$K{=}4$ rewrites per text). We use the exact implementation and preprocessing given in \cite{zhou2026learn} and use their rewrite data if we found it. 
However, if rewrite data is not available, we regenerate the write using the following rewrite prompt: \texttt{You are a rewriting expert and
you would rewrite the text without missing the original details.
Return ONLY the rewritten version. Do not explain changes, do not
give multiple options, and do not add commentary. Original text:.}

\subsection{Detection Methods}
\paragraph{L2D.}

For each text $T_i$, we generate $K=4$ rewrites $R_i^{(1)},\ldots,R_i^{(K)}$
using \texttt{gemma-9b-instruct}.
An \emph{AdaDist} scoring head, built on top of \texttt{gemma-9b-instruct},
is a model that is fine-tuned on 500 training samples using an AUC objective to produce
a scalar score $s_i = f(T_i, R_i)$ that is larger when $T_i$ is human (the original
scores higher than its rewrite) and smaller when $T_i$ is AI-generated. We use the model in \cite{chen2025imitate} and did not re-train the model.
At evaluation, the trained model assigns a signed score to each test text;
these scores are used directly as knockoff statistics $s_i$.

\paragraph{IMBD.}
IMBD reuses the same \texttt{gemma-9b-instruct} rewrites produced for L2D.
Instead of an AUC objective, the scoring model is fine-tuned via
\emph{Stochastic Preference Optimisation} (SPO; $\beta{=}0.05$,
lr$=10^{-4}$), which trains the model to prefer human texts over their
LLM rewrites. We generate the model for using the first half the data and apply it to the second half and vice-versa to generate scores independent of model.

After generating the score $g(\cdot)$ for $T_i$ and its rewrite $R_i$, we take the knockoff statistic to be 
$$s_i = g(T_i)- g(R_i)$$ 
which is then fed into the knockoff filter.
\paragraph{Likelihood.}
The likelihood method requires \emph{no training}.
For each text $T_i$ and its rewrite $R_i$ (again from \texttt{gemma-9b-instruct}),
we compute the signed statistic
\begin{equation}
  s_i \;=\; g(R_i) - g(T_i),
  \label{eq:likelihood_stat}
\end{equation}
where $g(\cdot)$ denotes the per-token log-likelihood under
\texttt{gemma-1b}. 
For human texts, the LLM rewrite $R_i$ achieves higher likelihood than the
original $T_i$, yielding $s_i > 0$.
For LLM texts, $T_i$ and $R_i$ are both samples from the same
distribution, so $s_i$ is approximately symmetric around zero---directly
satisfying Assumption~\ref{assumption:key_assumption}. We found that using \texttt{gemma-9b-instruct} as the likelihood model introduces a systematic directional bias in the null statistics, resulting in very skewed distribution even under null.

\subsection{Knockoff Filter}

After obtaining signed statistics $\{s_i\}$ from any of the three methods
above, we apply the knockoff filter identically across all methods.

\paragraph{Section~4.1 --- Symmetry condition.}
We verify Assumption~\ref{assumption:key_assumption} empirically by computing the
fraction of null (AI-text) statistics satisfying $s_i > 0$ (frac$+$) and
applying a Kolmogorov--Smirnov test for symmetry around zero.
Values of frac$+$ close to $0.5$ and large KS $p$-values indicate that the
assumption holds. We remove the mean $s_i$ from all the $s_i$ as a demean procedure.

\paragraph{Section~4.2 --- Empirical FDR control.}
We first center the statistics by subtracting their empirical mean as a calibration step, then apply the knockoff threshold procedure at target levels
$q \in \{ 0.2, 0.3, 0.5\}$.
A text is flagged as human if its statistic exceeds the data-adaptive
threshold $\tau(q)$; the procedure guarantees actual FDR $\leq q$ whenever
the symmetry assumption holds.
\newpage
\subsection{More Experiment results}

Figure \ref{fig:per_domain_all_q} disaggregates the empirical FDR and detection power from Table 2 across all 19 domains at each target level $q \in { 0.2, 0.3, 0.5}$. Each panel shows grouped bars for GPT-3.5T, GPT-4o, Gemini-1.5-Pro, and Llama-3-70B; the dashed horizontal line marks the nominal FDR target.

\begin{table*}[t]
\centering
\small
\begin{tabular}{lll}
\hline
\textbf{Category} & \textbf{Source} & \textbf{License} \\
\hline
Academic Research      & Arxiv abstracts \cite{maoraidar}                                           & Various CC licenses        \\
Art Culture           & Wikipedia                                                                   & CC BY-SA                   \\
Business              & Wikipedia                                                                   & CC BY-SA                   \\
Educational Material  & Ghostbuster essays \cite{verma2024ghostbuster}                              & CC BY 3.0                  \\
Entertainment         & IMDb dataset \cite{imdb2024}, Stanford SST2 \cite{socher-etal-2013-recursive}      & IMDb terms of use, CC Zero \\
Environmental         & Climate-Ins \cite{spokoyny2023towards}                                      & CC Zero                    \\
Finance               & Hugging Face FIQA \cite{thakur2021beir}                                     & CC BY-NC                   \\
Food                  & Kaggle fine food reviews \cite{mcauley2013amateurs}                         & CC Zero                    \\
Government            & Wikipedia                                                                   & CC BY-SA                   \\
Legal                 & CaseHOLD \cite{zheng2021does}                                               & Apache 2.0                 \\ 
Medical Text          & PubMedQA \cite{jin2019pubmedqa}                                             & MIT                        \\
News Article          & XSum \cite{narayan2018don}                                                 & MIT                        \\
Online Content        & Hugging Face blog authorship \cite{schler2006effects}                       & Non-commercial             \\ 
Product Review        & Yelp reviews \cite{maoraidar}                                           & Yelp terms of use          \\
Religious             & Bible, Buddha, Koran, Meditation, and Mormon                                & N/A                        \\
Sports                & Olympics website \cite{olympics2024}                                           & Olympics terms of use      \\
Technical Writing     & Scientific articles \cite{mosca2023distinguishing}                          & CC Zero                    \\
Travel Tourism        & Wikipedia                                                                   & CC BY-SA                   \\
\hline
\end{tabular}
\caption{Source and license for each of the 19 domains in our dataset.}
\label{tab:dataset_details}
\end{table*}
\begin{figure*}[ht]
  \centering
  \begin{subfigure}[t]{1\linewidth}
    \includegraphics[width=\linewidth]{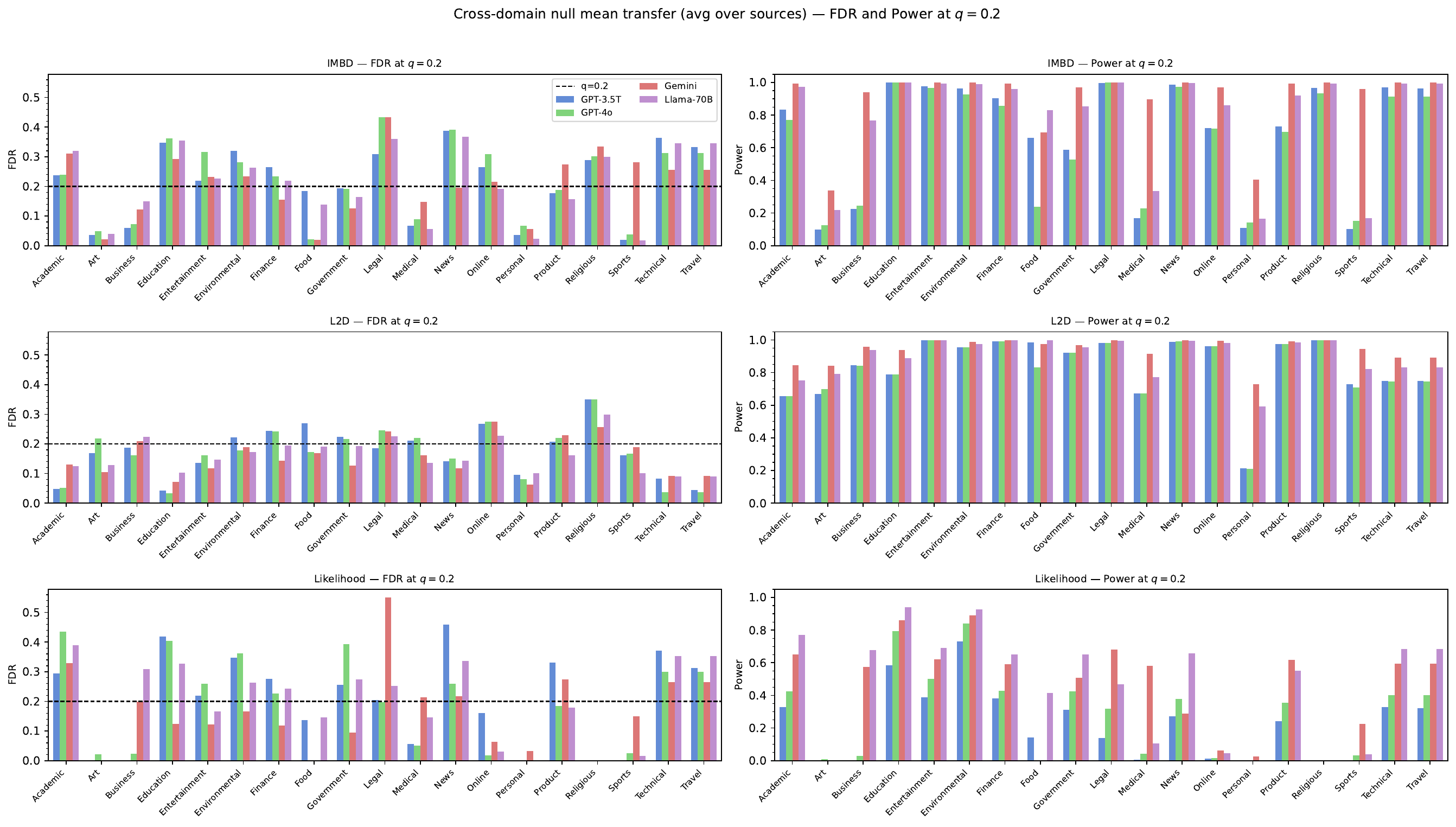}
    \caption{$q=0.2$}
\vspace{6pt}
  \end{subfigure}\hfill
  \begin{subfigure}[t]{1\linewidth}
    \includegraphics[width=\linewidth]{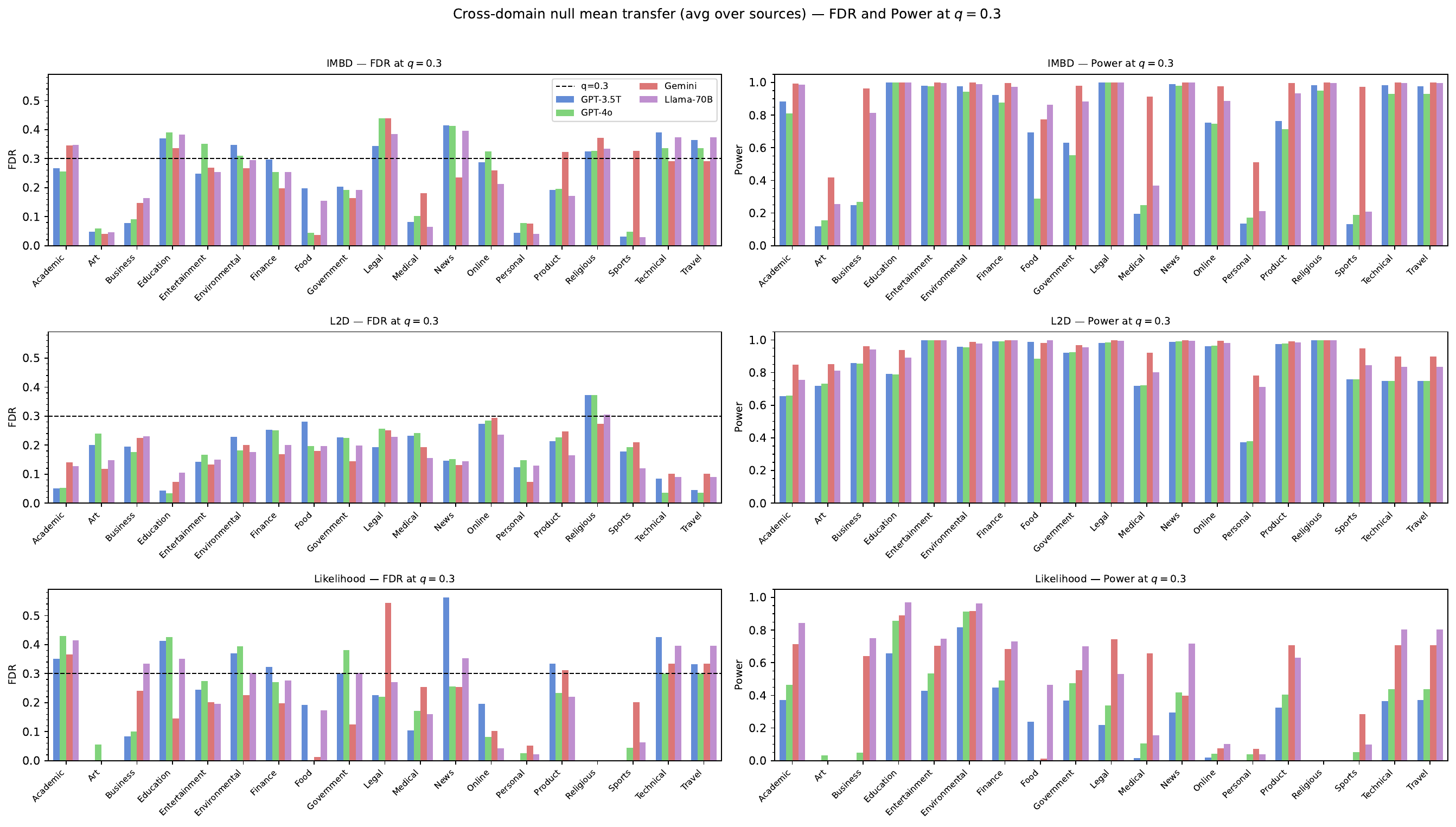}
    \caption{$q=0.3$}
  \end{subfigure}
\end{figure*}
  \vspace{6pt}
\begin{figure*}[ht]
  \centering
  \begin{subfigure}[t]{1\linewidth}
    \includegraphics[width=\linewidth]{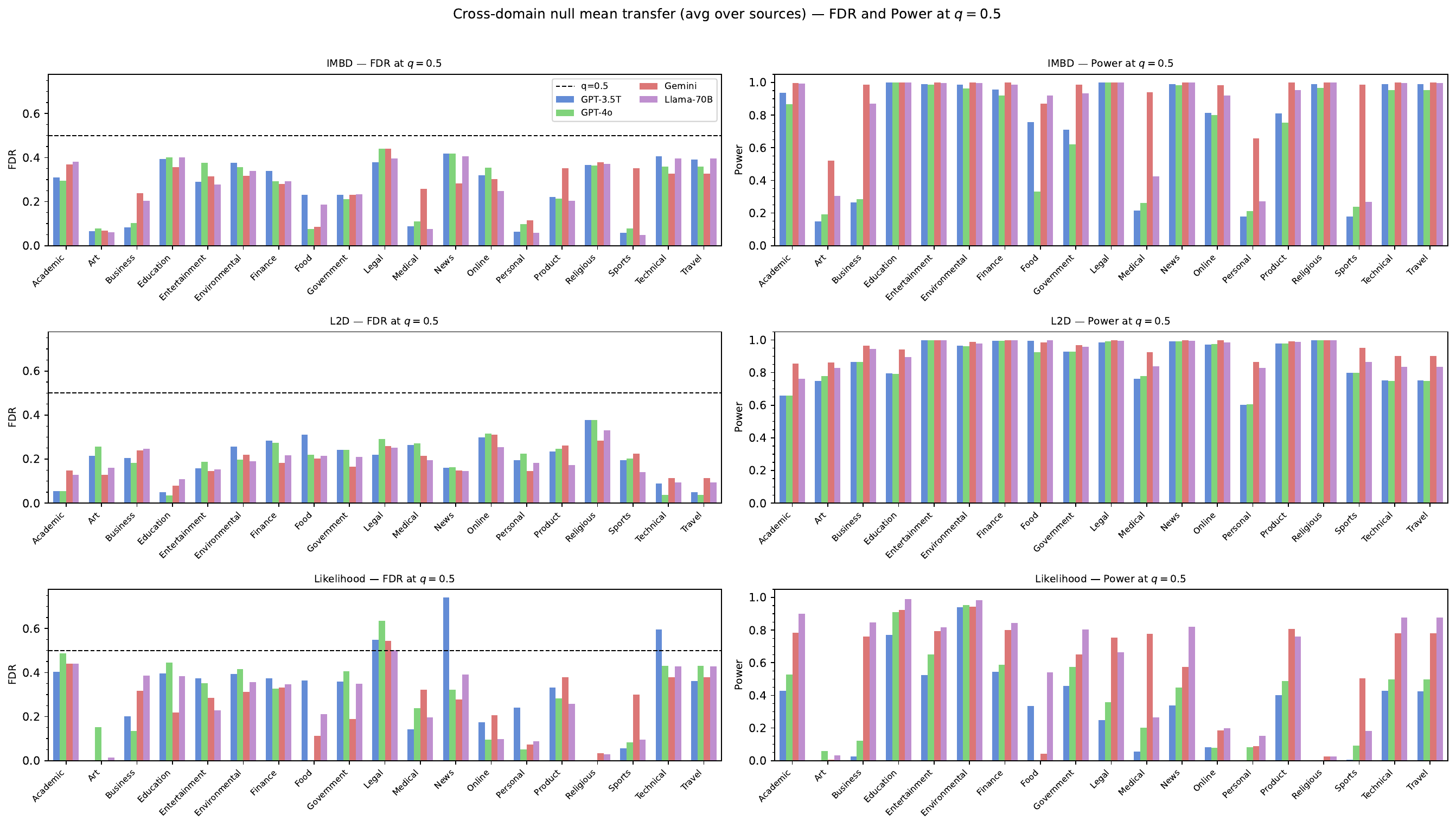}
    \caption{$q=0.5$}

  \end{subfigure}\hfill
  \caption{Per-domain FDR and Power across all methods and models at each
           target level $q \in \{ 0.2, 0.3, 0.5\}$.
           Dashed lines mark the corresponding target FDR.}
  \label{fig:per_domain_all_q}
\end{figure*}
\section{Use of AI}
We use AI for writing improvement and code debugging. AI is used to improve the writing of the text.

\end{document}